\documentclass[10pt,journal,compsoc]{IEEEtran}
\usepackage[utf8]{inputenc}
\usepackage[english]{babel} 
\usepackage[linesnumbered]{algorithm2e}
\usepackage{amsmath,amsfonts, amssymb} 
\usepackage{amsthm}
\usepackage{mathtools}
\usepackage{acronym}
\usepackage{longtable}
\usepackage{textcomp}
\usepackage{multirow}
\usepackage{makecell}
\usepackage{subfigure}
\usepackage{graphicx}
\usepackage{color}
\usepackage{epsfig}
\usepackage{url}
\usepackage{balance}
\usepackage{float}
\usepackage{comment}
\usepackage{pifont}
\usepackage{hyperref}
\usepackage{xr}
\usepackage{array}
\usepackage{cite}
\usepackage{tikzpeople}
\usepackage{tabularx}
\usepackage{bbding}

\usepackage[
	n,
	operators,
	advantage,
	sets,
	adversary,
	landau,
	probability,
	notions,	
	logic,
	ff,
	mm,
	primitives,
	events,
	complexity,
	asymptotics,
	keys]{cryptocode}

\usetikzlibrary{intersections, positioning, shapes.multipart}

\newcolumntype{P}[1]{>{\centering\arraybackslash}p{#1}}
\newcommand{\proto}{\emph{SpreadMeNot}}

\newtheorem{theorem}{Theorem}
\newtheorem{definition}{Definition}

\acrodef{API}{Application Program Interface}
\acrodef{BLE}{Bluetooth Low Energy}
\acrodef{CA}{Certification Authority}
\acrodef{SDR}{Software Defined Radio}
\acrodef{DoS}{Denial of Service}
\acrodef{DH}{Diffie Hellman}
\acrodef{ECC}{Elliptic Curve Cryptography}
\acrodef{ECDH}{Elliptic Curve Diffie Hellman}
\acrodef{ECDLP}{Elliptic Curve Discrete Logarithm Problem}
\acrodef{ECDSA}{Elliptic Curve Digital Signature Algorithm}
\acrodef{GNSS}{Global Navigation Satellite System}
\acrodef{GPS}{Global Positioning System}
\acrodef{IoT}{Internet of Things}
\acrodef{KDF}{Key Derivation Function}
\acrodef{HKDF}{HMAC Key Derivation Function}
\acrodef{MEO}{Medium Earth Orbit}
\acrodef{MITM}{Man-in-the-Middle}
\acrodef{PKC}{Public Key Cryptography}
\acrodef{PKI}{Public Key Infrastructure}
\acrodef{RF}{Radio Frequency}
\acrodef{RFID}{Radio Frequency Identification}
\acrodef{RSS}{Received Signal Strength}
\acrodef{SNR}{Signal-to-Noise-Ratio}
\acrodef{TLS}{Transport Layer Security}
\acrodef{PEPP-PT}{Pan-European Privacy-Preserving Proximity Tracing}
\acrodef{TCN}{Temporary Contact Numbers}
\acrodef{UHF}{Ultra High Frequency}
\acrodef{TDMA}{Time Division Multiple Access}
\acrodef{GFSK}{Gaussian Frequency Shift Keying}
\acrodef{DPSK}{Differential Phase Shift Keying}
\acrodef{FHSS}{Frequency-Hopping Spread Spectrum}
\acrodef{PACT}{Private Automated Contact Tracing}
\acrodef{DP-3T}{Decentralized Privacy-Preserving Proximity Tracing}
\acrodef{SIG}{Special Interest Group}
\acrodef{GDPR}{General Data Protection Regulation}
\acrodef{PRF}{Pseudorandom Function}
\graphicspath{{figs/}}

\begin{document}

\title{\proto: A Provably Secure and Privacy-Preserving Contact Tracing Protocol}
\author{
    \IEEEauthorblockN{Pietro Tedeschi, \IEEEmembership{Student Member,~IEEE,} Spiridon Bakiras, \IEEEmembership{Member,~IEEE,} and Roberto Di Pietro, \IEEEmembership{Senior Member,~IEEE}}\\
    \IEEEcompsocitemizethanks{\IEEEcompsocthanksitem Pietro Tedeschi, Spiridon Bakiras, and Roberto Di Pietro are with the Division of Information and Computing Technology (ICT), College of Science and Engineering (CSE), Hamad Bin Khalifa University (HBKU), Doha, Qatar. \protect\\
    e-mails: \{ptedeschi, sbakiras, rdipietro\}@hbku.edu.qa
    }
}

\maketitle

\begin{abstract}
A plethora of contact tracing apps have been developed and deployed in several countries around the world in the battle against \textsc{Covid-19}. However, people are rightfully concerned about the security and privacy risks of such applications. To this end, the contribution of this work is twofold. First, we present an in-depth analysis of the security and privacy characteristics of the most prominent contact tracing protocols, under both passive and active adversaries. The results of our study indicate that all protocols are vulnerable to a variety of attacks, mainly due to the deterministic nature of the underlying cryptographic protocols. Our second contribution is the design and implementation of \proto, a novel contact tracing protocol that can defend against most passive and active attacks, thus providing strong (provable) security and privacy guarantees that are necessary for such a sensitive application.

Our detailed analysis, both formal and experimental, shows that \proto~satisfies security, privacy, and performance requirements, hence being  an ideal candidate for building a contact tracing solution that can be adopted by the majority of the general public, as well as to serve as an open-source reference for further developments in the field. 
\end{abstract}

\begin{IEEEkeywords}
security, privacy, contact tracing, public-key cryptography
\end{IEEEkeywords}

\section{Introduction}
\label{sec:intro}
The sudden outbreak of the \textsc{Covid-19} coronavirus has fundamentally changed our society. The high infection rate of the virus facilitated its rapid spread around the world that resulted in the most deadly pandemic in recent history~\cite{wen2020study}. Unfortunately, according to health experts, \textsc{Covid-19} will almost certainly not be the last pandemic, so we, as a society, should be prepared to tackle future outbreaks more effectively and efficiently. 
To this end, contact tracing---followed by aggressive testing---has proven to be a very valuable tool in the battle against \textsc{Covid-19}~\cite{ferretti2020, martin2020}. In particular, contact tracing involves the early identification and notification of people that have been exposed to the virus by being in close proximity, for a given time in the near past, to a user that has tested positive.

Traditionally, contact tracing has been performed exclusively by public health professionals, in the form of interviews. However, relying solely on patient interviews is not a very effective approach. First, it requires an enormous workforce to trace potential infection chains when there are thousands of new cases discovered daily. Second, face-to-face or even phone interviews regarding social contacts may feel like an invasion of privacy to many people, who might refrain from sharing all the relevant information. Finally, close contact with complete strangers is a regular theme in our daily lives (public transit, supermarkets, shopping malls, etc.), and such infection chains cannot be tracked with manual contact tracing methods~\cite{tedeschi2020iotrace}.

Alternatively, \textit{digital} contact tracing is a technology that is gaining significant traction among governments and public health officials. In a nutshell, digital contact tracing is generally 
achieved via mobile applications that leverage either the \ac{BLE} protocol or the  \ac{GNSS} technologies to keep track of other mobile devices in their vicinity~\cite{cunche2020}. More specifically, every device periodically broadcasts a random beacon that is intercepted and recorded by the nearby devices. If a beacon is detected to be closer than a certain threshold (e.g., $2$ meters) for a significant amount of time, the event is recorded permanently in the \textit{contact list} maintained by that device. When a user tests positive for \textsc{Covid-19}, the public health authorities are given access to their device in order to release their random beacons (and/or contact list) to a centralized database. Subsequently, other devices will download the released information and identify whether they have been in contact with the infected individual.

There have been a plethora of contact tracing apps implemented and used around the world. Singapore's TraceTogether app~\cite{tracetogether} was one of the first wide-scale deployments, while many other big players have entered  the arena, including the European Union~\cite{pepppt, troncoso2020} and Apple/Google~\cite{applegoogle}. Nevertheless, people nowadays are very concerned about their privacy and may be unwilling to install and run a surveillance-type application~\cite{sun2020vetting, garg2020}. Further, the energy tool of such applications and the perceived battery life shortening could limit the adoption of contact tracing by the general public. A solution that on one hand offers provable security and privacy guarantees, while on the other hand introduces only a little overhead, is in dire need.

\textbf{Contribution.} 
To this end, the first contribution of this work is an in-depth analysis of the security and privacy characteristics of the most prominent contact tracing protocols. We look at the fundamental mechanisms incorporated in these methods, including the type of information that is collected, the location where that information is stored (centralized vs. decentralized vs. hybrid), and the cryptographic primitives involved in the generation of the random beacons. We also consider a wide range of adversarial behavior, ranging from simple passive (eavesdropping) attacks to more elaborate, active, ones such as replay and relay attacks.\\
Our results indicate that, at the very least, all existing approaches are vulnerable to simple eavesdropping attacks that can de-anonymize an individual once they have tested positive for the coronavirus. Indeed, beacons are generated in a deterministic manner (e.g., using hashing) and are always transmitted in clear text. As such, an adversary can easily intercept them with off-the-shelf antennas, and tag them with time-stamped location information. Therefore, when users disclose their beacons (or contact lists), the adversary can track their entire location history over the past two weeks.

This led us to our second contribution, which is the design of \proto, a novel contact tracing protocol with very strong, provable privacy guarantees that protects the users' proximity data against sophisticated attacks. The main novelty of our protocol is the use of beacons that are generated with public-key cryptographic primitives. The probabilistic nature of public-key cryptography allows for the randomization of all previously transmitted beacons, so that they can be safely published without disclosing any identifiable information to an adversary. As such, \proto\ is resilient to eavesdropping attacks. Furthermore, to strengthen our protocol against more powerful, active adversaries, we introduce a simple extension (based on digital signatures) that prevents an adversary from replaying previously transmitted beacons. Finally, we demonstrate with experiments\footnote{The source code of \proto~will be released as open source when the paper accepted} run on mobile phones that \proto's overhead in terms of computations and energy consumption is very reasonable for an average smartphone device. 

\textbf{Roadmap.} The remainder of the paper is organized as follows. Section~\ref{sec:background} introduces the technical background related to contact tracing technologies and public-key cryptography. Section~\ref{sec:scenario_adv_model} illustrates the generic contact tracing model and describes the adversarial behavior that we assume in this paper. Section~\ref{sec:privacy_analysis} presents our in-depth study on the security and privacy of the current state-of-the-art solutions. Section~\ref{sec:proto} introduces the \proto\ protocol and Section~\ref{sec:performance} investigates its performance in terms of computations and energy requirements. Section~\ref{sec:discussion} highlights some important optimizations that mitigate the cost of public-key cryptography, and Section~\ref{sec:conclusion} concludes our paper. 

\section{Technical Background}
\label{sec:background}
In this section, we discuss briefly the wireless technologies that are employed by existing contact tracing protocols, including Bluetooth and satellite communications. We also introduce \ac{ECC}, which is the underlying cryptographic primitive of \proto.
\noindent

\subsection{Bluetooth}
\label{sec:bluetooth}
The wireless Bluetooth technology was conceived in the early 1990s at Ericsson and was intended to replace the RS-232 data cable. Specifically, it allows fixed and mobile devices to exchange data over short distances using the \ac{UHF} spectrum. Bluetooth was first standardized as the IEEE 802.15.1 protocol, with the current standard maintained by the Bluetooth \ac{SIG}.

Nowadays, there is a wide range of applications that adopt Bluetooth technology to manage wireless devices such as headphones, mice, keyboards, and printers. More importantly, Bluetooth is the underlying technology for numerous novel applications, such as indoor localization, \ac{IoT}, contact tracing, gaming, smart-locking, networking, and data streaming, to name a few. For example, by placing Bluetooth transceivers around large shopping areas, we can enhance user experience via Bluetooth-based mobile advertising applications. Note that, under the current standard, a master Bluetooth device can establish a one-to-one communication with a maximum of $7$ devices in an ad-hoc network. 

From the physical layer perspective, Bluetooth communications adopt the \ac{UHF} frequency band and, in particular, the frequencies ranging from $2.402$~GHz to $2.480$~GHz. The modulation scheme is either \ac{GFSK} with a bit rate of $1$~Mbits/sec, or \ac{DPSK} with a maximum bit rate of $3$~Mbits/sec. Furthermore, Bluetooth employs \ac{FHSS} and divides the spectrum into $79$ channels, each with a bandwidth of $1$~MHz. 
\ac{BLE} accommodates $40$ channels, each with a bandwidth of $2$~MHz. 
The transmission range depends on the device class and its maximum transmission power, as shown in the Table~\ref{tab:bl_dev_class}.

\begin{table}[htbp]
\color{black}
\caption{Bluetooth Device Class.
}
\centering
\begin{tabular}{|c|c|c|c|c|}
\hline
\textbf{Device Class} & \textbf{Max TX Power} ($dBm$) & \textbf{Range} ($m$) \\ \hline\hline
$1$ & $20$ & $~100$ \\ \hline
$1.5$ & $10$ & $~20$ \\ \hline
$2$ & $4$ & $~10$ \\ \hline
$3$ & $0$ & $~1$ \\ \hline
$4$ & $-3$ & $~0.5$ \\ \hline
\end{tabular}
\label{tab:bl_dev_class}
\end{table}

The Bluetooth \ac{SIG} has released several versions of the Bluetooth standard, each supporting backward compatibility. The latest release is the Bluetooth Core Specification version 5.2. Notice that, \ac{BLE} was first introduced in version 4.0 and was further improved in versions 4.1 and 4.2. In this work, we consider the Bluetooth Core Specification version 4.2 standard. Under the cited  standard, the modulation rate is $1$~Mbits/sec and the format of the Bluetooth frame is depicted in Fig.~\ref{fig:ble_frame}.

\begin{figure}[htbp]
  \centering
  \includegraphics[angle=0, width=\columnwidth]{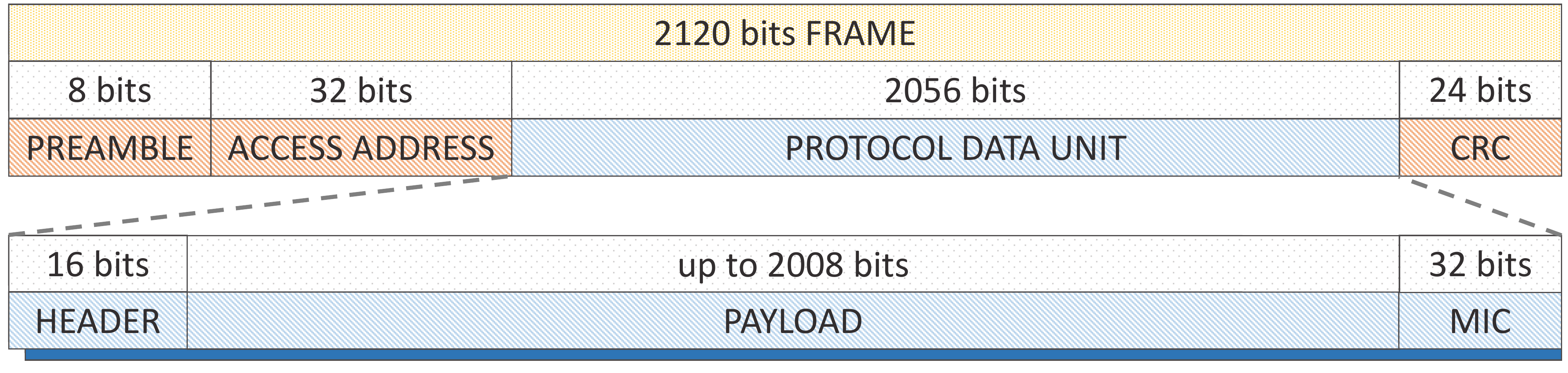}
  \caption{Bluetooth v4.2 frame: Packet format (upper part) and Protocol Data Unit (bottom part).}
  \label{fig:ble_frame}
\end{figure}

Bluetooth v4.2 frames have an overall size of $2120$~bits~\cite{bluetooth2010inc}. The message starts with the first $8$~bits reserved to a \textit{preamble} that is used to identify an upcoming Bluetooth message and allow the receiver to synchronize with the symbols emitted by the transmitter. 
An \emph{Access Address} of $32$~bits denotes the correlation code tuned to the physical channel. The \textit{Protocol Data Unit} (PDU) consists of $2056$~bits, and is reserved for Data TX and Advertising. Finally, the last $24$~bits are devoted to an error-detection code (CRC) and store the checksum of all bytes in the PDU. The bottom part of Fig.~\ref{fig:ble_frame} shows the contents of the PDU. It includes the Bluetooth packet \emph{Header} ($16$~bits), the link-layer \emph{Payload} (up to $2008$~bits of data), and $32$~bits of a \textit{Message Integrity Check} (MIC) \cite{gupta2016inside}.

Note that Bluetooth RF operations take place according to a slot-based \ac{TDMA} schedule. Specifically, assuming a data rate of $1$~Mbits/sec and a packet size of $265$~bytes, the transfer duration is $2.12$~ms.

\subsection{\acl{GNSS}}
\label{sec:gnss}
Some proximity tracing solutions adopt \ac{GNSS} technologies as a key element to approximate user location. Any smart device equipped with a \ac{GNSS} module can receive RF signals originating from \ac{MEO} satellites, located $19,000$ to $23,000$~km above Earth. Common \ac{GNSS} devices are only able to receive in the Upper L-Band, i.e., in the $1.5$ GHz range, with an
accuracy of about $3-5$ meters for the positioning \cite{gps_sps_ps}.
Each satellite is synchronized with atomic clocks and sends a navigation signal that contains information about the delivery time and the deviation from its expected trajectory.

Several \ac{GNSS} technologies are available on the market, with the most popular one being \ac{GPS} that is operated by the United States Department of Defense. Other systems include the Russian GLONASS, the European GALILEO, and the Chinese BEIDOU. Nevertheless, \ac{GNSS} technologies have some well-known weaknesses: they do not work in indoor environments,  are extremely sensitive to jamming attacks, and  are vulnerable to spoofing attacks if the RF signal is not encrypted~\cite{Schmidt2016}.

\ac{GNSS} technologies are widely adopted for several applications that span from location-based services to the monitoring and tracking of objects in remote areas. However, in the proximity tracing scenario, and compared to Bluetooth-based solutions, \ac{GNSS} has the following drawbacks: (i)  higher energy consumption due to the position sensors; (ii)  lower accuracy in terms of proximity localization; and, (iii) an inability to work in close spaces, like buildings.

\subsection{Elliptic Curve Cryptography}
\label{sec:ecc}
\acl{ECC} is a powerful approach to public-key cryptography that is based on the algebraic structure of elliptic curves over finite fields. \ac{ECC} is considered as a better alternative to traditional schemes over finite fields because, for the same level of security, it offers much smaller keys and ciphertexts~\cite{darrell_ecc}. 
As such, \ac{ECC} is more suitable for resource-constrained devices, such as smartphones or \ac{IoT} devices. An elliptic curve defines a set of coordinates $(x,y)$ that satisfy Eq.~\ref{eq:ecc_eq} below. In addition, a special point $\mathcal{O}$, namely point at infinity, represents the point at the ends of all lines parallel to the $y$-axis. 
\begin{equation}
    \label{eq:ecc_eq}
    y^2 = x^3 + ax + b
\end{equation}
The elliptic curve domain parameters over the finite field $\mathbb{F}_p$ are a sextuple $\mathcal{E} =(p, a, b, G, q, h)$, where: (i) $p>3$ is an integer specifying the finite field $\mathbb{F}_p$; (ii) $a, b$ are the elements that define the elliptic curve; (iii) $G \in \mathcal{E}(\mathbb{F}_p)$ is the generator point; (iv) prime $q$ is the order of the subgroup generated by $G$; and, (v) integer $h$ is the cofactor of the subgroup.


In \ac{ECC}, the private key is generated by choosing a uniformly random number $x\in \ZZ_q^*$, i.e., a number in the interval $[1,q)$. The corresponding public key is then computed as $P=xG$. The security of \ac{ECC} is based on the intractability of finding the discrete logarithm of a random elliptic curve point $P$ with respect to a publicly known base point $G$. This is known as the \ac{ECDLP}. The bit-length $n$ of the prime order $q$ determines the security level of an \ac{ECC} instantiation. We formally define the assumption behind the security of \ac{ECC} schemes below.

\begin{definition}
A function $f$ is \textbf{negligible} if for every polynomial $p(\cdot)$ there exists $N$ such that, for all integers $n>N$, it holds that $f(n) < \frac{1}{p(n)}$. We denote a negligible function of $n$ as $\negl$.
\end{definition}

\begin{definition}
The \textbf{\ac{ECDLP} assumption} is as follows: given $P, G \in \mathcal{E}(\mathbb{F}_p)$ where $P = xG$, $x \in \ZZ_q$, and $q \in \{0,1\}^n$ is the order of the elements $P$ and $G$, a polynomial-time algorithm can output $x$ with probability $\negl$. 
\end{definition}



\section{System and Adversarial Model}
\label{sec:scenario_adv_model}
In this section, we introduce the generic contact tracing environment assumed in our work (Section~\ref{sec:scenario}) and present the details of the underlying adversarial model (Section~\ref{sec:adv_model}).

\subsection{System Model}
\label{sec:scenario}
We assume a typical \ac{BLE}-based contact tracing environment, as illustrated in Fig.~\ref{fig:scenario}. Users that have subscribed to the contact tracing system are constantly broadcasting ephemeral identifiers (or beacons\footnote{Henceforth, we use the terms beacon and ephemeral ID interchangeably.}) that are randomly generated according to \proto's protocol specifications. Additionally, all users in proximity to a broadcasted beacon will temporarily store it into their local memory.

\begin{figure}[ht!]
\centering
\includegraphics[width=\columnwidth]{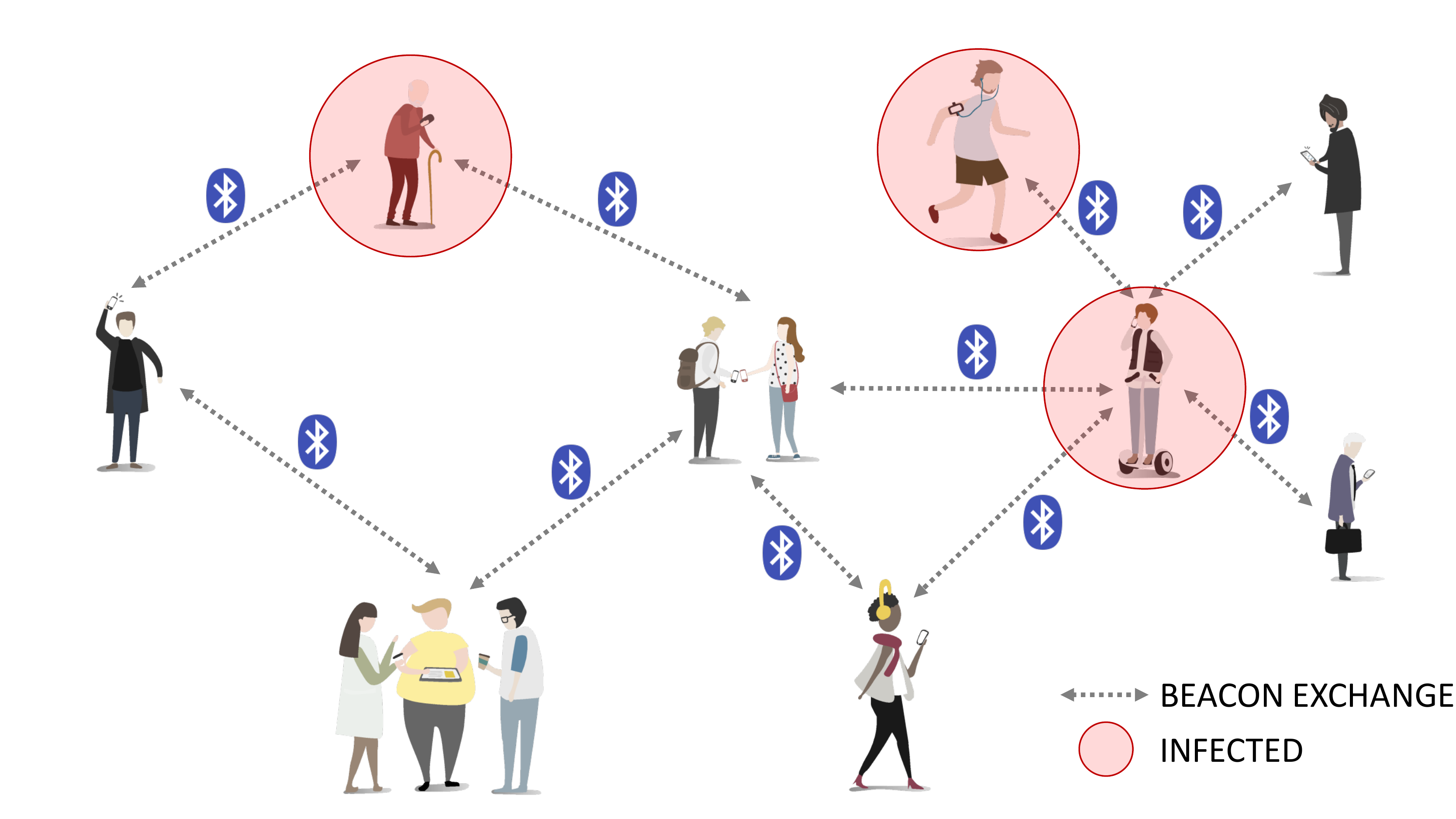}
\caption{\ac{BLE}-based contact tracing.}
\label{fig:scenario}
\end{figure}

A newly created beacon is valid for a fixed time window, in order to allow the surrounding devices to detect possible contagion events. In most existing applications, beacons are changing every $10$ to $15$ minutes, thus preventing malicious adversaries from tracking individual users. When a device receives the same beacon over a period of time that is considered significant by the public health authorities (and if the device is assumed to be sufficiently close to the beacon's source) the beacon is permanently stored in the \textit{contact list} of the device. Based on the characteristics of the targeted virus, contact events can be safely erased after a few days (e.g., $14$ days for \textsc{Covid-19}).

When a user is diagnosed as positive (e.g., the highlighted users in Fig.~\ref{fig:scenario}), the public health authorities will coordinate with the user to upload the stored contact list to their centralized database. The remaining users will periodically download the new contact lists from the database and check whether their own beacons are contained therein. If this is true, the user will be notified about the verified contagion event and will be given instructions for further actions.

\subsection{Adversarial Model}
\label{sec:adv_model}
The adversary assumed in our work is very powerful and may perform both passive and active attacks. We assume that the adversary is equipped with a powerful antenna, which can be either a regular Bluetooth handheld device, or a \ac{SDR} that is operated through a laptop/smartphone running an \ac{SDR}-compatible software tool, such as GNURadio~\cite{Tuttlebee2002}.

\subsubsection{Passive Attacks}
\hfill\\
\indent \textbf{Eavesdropping.} We assume that the adversary is a global eavesdropper, able to detect and decode any message broadcasted on the Bluetooth communication channel. This capability allows the adversary to obtain all the information  within the exchanged Bluetooth beacons, including the ``rolling'' standard ID adopted by Bluetooth technology and the ephemeral IDs of the underlying contact tracing protocol. In addition, the adversary locally enriches the intercepted beacon with appropriate metadata, such as time and location information. The adversary could hence conduct a cross-referencing analysis with data collected from various sources, including the ``infected'' beacons published by the public health authorities, background information on individual users, direct observation of user movements, etc.. 
The goals of an eavesdropping attack can be manifold. For example, the adversary may try to de-anonymize users that have tested positive for the virus, track user locations over time, or identify social relationships among users.

\subsubsection{Active Attacks}
\hfill\\
\indent \textbf{Replay.} We also assume that the adversary has Bluetooth transmission capabilities. Thus, it can replay beacons previously acquired  (via eavesdropping) from the Bluetooth communication channel. These beacons contain the ephemeral IDs of legitimate users and are, therefore, indistinguishable from legitimate beacons when processed by the contact tracing app. The goal of the adversary is to inject false contact events into the users' contact lists, 
hence increasing the chances for the user to falsely being considered a contagion risk---consequences could be, for instance, a mandatory, legally binding, period of quarantine for the attack's victim.

\textbf{Relay.} This attack aims at bypassing replay attack countermeasures that may incorporate timing information within the beacons, thus limiting the time window where beacons can be replayed. It can be thought of as an amplified version of a replay attack, where the adversary instantaneously disseminates every captured beacon to multiple locations under its control. The beacons are then replayed immediately within each of those locations. 

The reason why replay/relay attacks are dangerous in the context of contact tracing is that they can be used to manipulate the recorded contact lists. Specifically, an adversary can install a powerful antenna at a busy location (e.g., a shopping mall) and start eavesdropping on the transmitted beacons. Then, it replays those beacons with a high transmit power so that they register as legitimate contacts to a large number of users (even if they are further away from the antenna). As a result, if one of these users has a positive diagnosis, there will be an overwhelming amount of false negative alerts that would saturate test centers (if suspect-positive users are tested) or likely result in mandatory quarantine if testing is not possible---hence imposing an unnecessary privation of personal freedom.

Note that, in this work, we do not address \ac{DoS} attacks that may be launched by an adversary to disrupt the operation of the contact tracing network. For example, the adversary might  try to jam the Bluetooth communication channel or compromise the availability of the centralized database. 
We also assume that the disclosure of contact lists from infected users is done securely by the public health authorities. In other words, an adversary is not able to inject false data into the centralized database. Such attacks are independent from 
the underlying contact tracing protocol, and hence considered out of scope with respect to our contributions.
\section{Security and Privacy Analysis of Existing Solutions}
\label{sec:privacy_analysis}
In this section, we present an in-depth analysis of the security and privacy characteristics of the current state-of-the-art contact tracing protocols. We focus our discussion on the following metrics.

\textbf{Architecture (C/D/H).} In a centralized (C) architecture, all mobile devices forward their proximity data to a centralized database that is administered by the public health authorities. 
Conversely, in a decentralized (D) approach, the mobile devices do not share their data with the authorities, unless they test positive for the virus. Furthermore, the contact tracing operation is performed in a fully decentralized manner at the individual devices. Finally, in a hybrid (H) architecture, data collection follows the decentralized approach (nothing is disclosed to the authorities), while the contact tracing operation is performed at a centralized location---by having the infected individuals reveal their entire contact history to the health authorities. Clearly, the centralized (and to a lesser extent the hybrid) architecture requires users to fully trust the authorities from a security and privacy perspective.

\textbf{Privacy from Authority.} 
This is a measure of the privacy that users enjoy with regards to the information that is disclosed to the public health authorities by the contact tracing application. For instance, in the absence of a fully decentralized architecture, a central authority can harm the privacy of the users by collecting sensitive data (e.g., \ac{GPS} locations, contact lists) and/or performing statistical inference on the acquired data.


\textbf{Linkage Attack Resilience.} A linkage attack attempts to identify users by analyzing an anonymized dataset, while cross-referencing it with external data sources, such as user-specific information, geolocation data, etc.~\cite{Gvili2020}. 

For instance, the beacons released to the authorities by a user that had a positive test can potentially reveal his/her identity if the beacons can be linked to precise geographic locations.



\textbf{Replay/Relay Attack Mitigation.} As explained previously, replay and relay attacks can be very damaging to the effectiveness of contact tracing. Therefore, contact tracing protocols must incorporate appropriate cryptographic mechanisms to limit their scope or, if possible, eliminate them altogether.



\textbf{Robustness against Eavesdropping.} This metric quantifies the usefulness of the beacons that are collected by an eavesdropping adversary. For example, if users upload their previously transmitted beacons to a public database (when deemed infected), an adversary can leverage them to de-anonymize the users.

\textbf{Ephemeral Identity.} An ephemeral identity is the temporary ID associated with the real user identity and is contained inside the transmitted beacons. This is the minimum privacy guarantee that all contact tracing protocols should provide. Additionally, ephemeral IDs should change frequently, to prevent the tracking of user movements.



\textbf{Exposure of Geolocation Data.} Some contact tracing protocols collect \ac{GPS} data to perform the proximity testing operation. As such, when users are diagnosed as positive, their recent trajectories are disclosed to the public health authorities and, subsequently, to the general public. The cited  approach is less private than BLE-based contact tracing, because it does not require from the adversary any effort to learn the detailed trajectory data.



\textbf{Open Source.} The protocol and the code are released as open source to facilitate thorough security and privacy audits from both the industry and academia.

\subsection{Existing Solutions}

\textbf{\ac{DP-3T}~\cite{troncoso2020}.} The \ac{DP-3T} protocol is the product of a large consortium of European universities and research institutes. It is a fully decentralized protocol that leverages \ac{BLE} beacons to exchange ephemeral IDs among the participating devices. As reported by Vaudenay~\cite{Vaudenay2020}, the ephemeral identifiers are generated with an \pcalgostyle{HMAC}--\pcalgostyle{SHA256} key, encrypted with \pcalgostyle{AES}--\pcalgostyle{CTR} or \pcalgostyle{Salsa20}. Further, these keys are kept in the device's memory and erased after they are no longer required by the health authorities (e.g., after $14$ days). The keys are derived from a root key that is randomly generated every day. The analysis of Vaudenay highlights several security and privacy issues, including a replay attack which is possible due to the absence of message authentication codes. When a user is diagnosed as positive, \ac{DP-3T} offers several options for disclosing their own beacons to the health authorities. The least private approach discloses all the past daily keys, while the more privacy-preserving option allows the user to choose which beacons to reveal. 

\textbf{Apple/Google~\cite{applegoogle}.} Apple and Google jointly designed and implemented a decentralized proximity tracing protocol that is very similar to \ac{DP-3T}. Unlike other approaches, this protocol is fully integrated within the operating system, and exposes appropriate APIs to third-party developers in order to develop their own smartphone applications. From a privacy perspective, Apple and Google have not released their source code, so users must trust these companies (and their operating systems) with respect to data analysis~\cite{Gvili2020}. With regards to the cryptographic specifications, the key schedule for contact tracing is organized into $3$ main phases: (i) tracing key generation, (ii) daily tracing key generation, and (iii) rolling proximity identifier generation. The tracing key generation procedure is executed when contact tracing is first enabled on the mobile device. The $32$-byte tracing key is derived via a Cryptographic Random Number Generator and is stored on the device. From this tracing key, an \ac{HKDF}, such as \pcalgostyle{HMAC}--\pcalgostyle{SHA256}, is employed to derive a $16$-byte daily tracing key that is valid for a $24$-hour time window, based on Unix epoch time. Finally, the rolling proximity identifiers exchanged within the \ac{BLE} beacons are $16$-byte ephemeral IDs, derived from the daily tracing key through an \pcalgostyle{HMAC}--\pcalgostyle{SHA256} operation~\cite{applegoogle_crypto_spec}.

\textbf{Berke \textit{et al.}~\cite{berke2020assessing}.} This is a \ac{GPS}-based solution that also adheres to the decentralized architecture. Specifically, the mobile devices are constantly monitoring the users' location, by recording their \ac{GPS} coordinates every $t$ minutes. Therefore, each point in the user's location history has three dimensions, namely longitude, latitude, and time. When a user is diagnosed as infected, their location history is uploaded to the health authorities' server, thus allowing other devices to download it and look for possible contagion events. To preserve privacy, precise \ac{GPS} coordinates are replaced by larger geographic areas, and the resulting point-intervals are obfuscated with a hash function like \pcalgostyle{SHA256}. In addition, the authors propose a technique where the server and individual clients invoke a private set intersection protocol to identify contagion events, without disclosing the patient's location history. In terms of security, an adversary can either broadcast fake \ac{GPS} signals~\cite{oligeri_wisec_19} (\ac{GPS} spoofing) or perform a \ac{GPS} jamming attack~\cite{Hu2009}, in order to disrupt the proximity tracing operation.

\textbf{Hamagen~\cite{Hamagen}.} Hamagen is a \textsc{Covid-19} exposure prevention app that was developed by Israel's Ministry of Health. It is fully decentralized and leverages \ac{GPS} measurements to identify possible \textsc{Covid-19} exposure events. Once an hour, the installed application downloads a file with an anonymous list of locations, times, and dates that confirmed (from the Ministry) patients have visited in the past. The app will then perform a cross-reference on the (timestamped) locations visited by the user and determine whether the user is at risk of exposure. Note that, the application exploits \ac{GPS}, Bluetooth, and WiFi to improve the accuracy of geolocation. Additionally, it does not rely on any cryptographic primitives.

\textbf{Corona 100m~\cite{corona100m}.} This is South Korea's centralized and \ac{GPS}-based solution that publicly informs citizens (and the health authorities) of known cases within $100$ meters of where they are. It is worth noticing that this application cross-references credit card data, medical records from hospitals, social networks graphs, and video surveillance footage, thus raising serious privacy concerns among the citizens.

\textbf{\ac{PEPP-PT}~\cite{pepppt}.} The \ac{PEPP-PT} protocol is the European Union's solution to digital contact tracing. \ac{PEPP-PT} is a hybrid approach based on \ac{BLE} technology and is compliant with the EU's \ac{GDPR} privacy laws. Unlike DP-3T and Apple/Google, \ac{PEPP-PT} mandates that the users' \ac{BLE} beacons (ephemeral IDs) are generated and distributed by the public health authorities. Specifically, during registration, each user is assigned a random ID by the system. Then, for each contact tracing epoch $t$, the system selects a global secret key $K_t$ and encrypts (using \pcalgostyle{AES}) all user IDs with the same key. The resulting ciphertexts are the ephemeral IDs that are distributed in advance to the individual users. When testing positive for the coronavirus, users release their logged contact lists to the authorities, who are then able to perform the contact tracing operation at the centralized server, since they can identify the IDs of all contacts by decrypting the beacons with the secret keys. Clearly, this approach is less privacy-preserving than the decentralized solutions, because it allows anyone who owns the \pcalgostyle{AES} keys to track all registered citizens.

\textbf{\ac{TCN}~\cite{tcn}.} \ac{TCN} is a fully decentralized protocol proposed by the Coalition Network. It leverages \ac{BLE} communications and its operation is very similar to \ac{DP-3T} and Apple/Google. In other words, mobile devices generate their own ephemeral IDs that are broadcasted to other devices in their proximity. Each device gradually builds its own private contact list that is not shared with the authorities. After a positive diagnosis, users upload their keys to the centralized database, in order to allow other devices to reconstruct the corresponding ephemeral IDs and measure their exposure risk. The main difference in \ac{TCN} is that the device generates a public and private key pair (based on \pcalgostyle{Ed25519} curves) so that other devices can verify the authenticity of the reconstructed ephemeral IDs. In particular, the public-key is used as an input to the \pcalgostyle{SHA256} hash function that generates the ephemeral IDs, while the private key is used to sign the report that is sent to the health authorities after the user is diagnosed as positive.

\textbf{BlueTrace~\cite{bluetrace}.} BlueTrace was one of the first protocols to enjoy a wide-scale deployment since it was the fundamental building block in Singapore's TraceTogether app. Its operation is quite similar to \ac{PEPP-PT}, i.e., a hybrid approach where the ephemeral IDs are generated by the public health authorities. As such, the contact tracing operation is performed in a centralized manner by the health authorities, after the positively diagnosed individuals upload their contact logs to the centralized server. The system maintains a global secret key for the generation of the ephemeral IDs and, during registration, every user is assigned a unique ID. Then, the per-user ephemeral IDs are generated by encrypting a message---comprising the user ID, the validity period (start and expiration), an initialization vector, and an authentication tag---with an \pcalgostyle{AES256}--\pcalgostyle{GCM} cipher and \pcalgostyle{Base64} encoding the resulting ciphertext. To minimize the effect of replay attacks, the validity period of each ephemeral ID is set to $15$ minutes~\cite{bluetrace}.

\textbf{\ac{PACT}~\cite{pact}.} \ac{PACT} is another contact tracing protocol that follows the decentralized approach of DP-3T and Apple/Google. In particular, every device generates its own ephemeral IDs through a \ac{PRF} that takes as input a random $256$-bit seed (which changes every hour) and the current time measured at one-minute precision. Furthermore, a newly diagnosed patient will upload all the random seeds from the recent past to the centralized database, in order to allow other users to reconstruct the ephemeral IDs and estimate their exposure risk. Note that, similar to other approaches, \ac{PACT} reduces the risk of replay attacks by including timestamps into the generation of the ephemeral IDs. 

\textbf{Whisper~\cite{whisper}.} Whisper is the only proximity tracing protocol thus far in the literature that employs public-key cryptography. It is a fully distributed protocol that leverages \ac{BLE} communications to monitor and log contact events. Specifically, every device periodically generates a fresh public/private key pair, based on \pcalgostyle{Ed25519} elliptic curves. Additionally, unlike other protocols, Whisper operates in two phases. In the first phase, the device periodically scans the \ac{BLE} channel to detect other Whisper devices in its proximity. For each discovered device, the protocol initiates a connection (phase two), during which the two devices perform a Diffie-Hellman key exchange to generate a shared key $K$. For each key $K$, the device computes two hash digests using the \pcalgostyle{BLAKE2}--\pcalgostyle{160} hash function: the \textit{tell-token} is the hash digest of $K$ and the device's own public-key, while the \textit{hear-token} is the hash digest of $K$ and the connecting peer's public-key. After a positive diagnosis, the device uploads its own tell-tokens to the centralized database, which enables other devices to compare them against their own hear-tokens. The major advantage of Whisper in terms of privacy is that the published tell-tokens do not reveal any information to an adversary, because it is infeasible to compute the underlying shared keys.

\textbf{IoTrace~\cite{tedeschi2020iotrace}.} IoTrace is a novel IoT-enabled approach to contact tracing.
Specifically, in IoTrace, mobile devices neither  receive the broadcasted beacons nor  maintain individual contact lists. Instead, the deployed IoT devices collect every beacon that is transmitted around them and send all data to a centralized server. When a user tests positive for the coronavirus, the mobile device releases its transmitted beacons to the authorities who then publish all the beacons that are considered in close proximity to the user. The remaining users periodically download the updated database and perform the exposure notification function in a distributed manner. The key characteristic of IoTrace is the reduced energy cost on mobile devices, due to the absence of frequent scanning operations to identify transmitted beacons. However, in this contribution, we only focus on peer-to-peer contact tracing that is performed without the use of an existing IoT infrastructure. As such, we do not include IoTrace in our comparison.
\\\\
\textcolor{black}{
A qualitative survey of the inherent privacy limitations of contact racing applications is provided in \cite{inherent_privacy_limitations}, whereas a quantitative comparison among different solutions and applications for contact tracing, such as PEPP-PT~\cite{pepppt}, DP-3T~\cite{troncoso2020}, Apple/Google~\cite{applegoogle}, Hamagen~\cite{Hamagen} and BlueTrace~\cite{bluetrace} is presented in \cite{tedeschi2020iotrace}. The analysis concerned the adopted wireless technology, the deployed architecture, the RF energy consumption, and the security and privacy aspects of elements such as location and health status. The analysis is further enriched by taking into account the storage requirements and the computational costs incurred by the adoption of given cryptographic primitives. 
Further, the authors experimentally assessed the performance by considering the Bluetooth SoC nRF51822 and the GPS SiP nRF9160 (for Hamagen) hardware platforms. 
}


\subsection{Security and Privacy Analysis}
In this section, we discuss the security and privacy characteristics of the aforementioned protocols. Our results are summarized in Table~\ref{tab:solutions_comparison}.

\begin{table*}[ht]
\caption{Comparison of our solution against the state-of-the-art approaches. None: \FiveStarOpen\FiveStarOpen\FiveStarOpen, Low: \FiveStar\FiveStarOpen\FiveStarOpen, Medium: \FiveStar\FiveStar\FiveStarOpen, High: \FiveStar\FiveStar\FiveStar
}
\centering
    \resizebox{\textwidth}{!}
    {\begin{tabular}{|p{2.8cm}||P{1.0cm}|P{0.8cm}|P{0.9cm}|P{1.0cm}|P{1.0cm}|P{0.8cm}|P{0.7cm}|P{1.1cm}|P{0.7cm}|P{1.2cm}|P{1.5cm}|}
    \hline
        \textbf{Features} & DP-3T \cite{troncoso2020} & A/G \cite{applegoogle} & Berke \textit{et al.} \cite{berke2020assessing} & Hamagen \cite{Hamagen} & Corona 100m\cite{corona100m} & PEPP-PT~\cite{pepppt} & TCN \cite{tcn} & BlueTrace \cite{bluetrace} & PACT \cite{pact} & Whisper \cite{whisper} & \proto \\ \hline\hline
        \emph{Architecture (C/D/H)} & D & D & D & D & C & H & D & H & D & D & D\\ \hline
        \emph{Wireless Technology} & BLE & BLE & GPS & GPS & GPS & BLE & BLE & BLE & BLE & BLE & BLE\\ \hline
        \emph{Privacy from Authority} & \FiveStar\FiveStar\FiveStarOpen & \FiveStar\FiveStar\FiveStarOpen & \FiveStar\FiveStar\FiveStarOpen & \FiveStar\FiveStar\FiveStarOpen & \FiveStarOpen\FiveStarOpen\FiveStarOpen & \FiveStarOpen\FiveStarOpen\FiveStarOpen & \FiveStar\FiveStar\FiveStarOpen & \FiveStarOpen\FiveStarOpen\FiveStarOpen & \FiveStar\FiveStar\FiveStarOpen & \FiveStar\FiveStar\FiveStar & \FiveStar\FiveStar\FiveStar\\ \hline
        \emph{Linkage Attack \newline Resilience} & \FiveStar\FiveStarOpen\FiveStarOpen & \FiveStar\FiveStarOpen\FiveStarOpen & \FiveStar\FiveStarOpen\FiveStarOpen & \FiveStar\FiveStarOpen\FiveStarOpen & \FiveStarOpen\FiveStarOpen\FiveStarOpen & \FiveStarOpen\FiveStarOpen\FiveStarOpen & \FiveStar\FiveStarOpen\FiveStarOpen & \FiveStarOpen\FiveStarOpen\FiveStarOpen & \FiveStar\FiveStarOpen\FiveStarOpen & \FiveStar\FiveStarOpen\FiveStarOpen & \FiveStar\FiveStar\FiveStar\\ \hline
        \emph{Replay/Relay \newline Attack Mitigation} & \FiveStar\FiveStarOpen\FiveStarOpen & \FiveStar\FiveStarOpen\FiveStarOpen & N/A & N/A & N/A & \FiveStar\FiveStarOpen\FiveStarOpen & \FiveStar\FiveStarOpen\FiveStarOpen & \FiveStar\FiveStarOpen\FiveStarOpen & \FiveStar\FiveStarOpen\FiveStarOpen & \FiveStar\FiveStarOpen\FiveStarOpen & \FiveStar\FiveStar\FiveStar \\ \hline
        \emph{Robustness against \newline Eavesdropping \hfill} & \FiveStar\FiveStar\FiveStarOpen & \FiveStar\FiveStar\FiveStarOpen & N/A & N/A & N/A & \FiveStarOpen\FiveStarOpen\FiveStarOpen & \FiveStar\FiveStar\FiveStarOpen & \FiveStarOpen\FiveStarOpen\FiveStarOpen & \FiveStar\FiveStar\FiveStarOpen & \FiveStar\FiveStar\FiveStar & \FiveStar\FiveStar\FiveStar \\ \hline
        \emph{Ephemeral Identity} & Yes & Yes & N/A & N/A & N/A & Yes & Yes & Yes & Yes & Yes & Yes \\ \hline
        \emph{Exposure of \newline Geolocation Data} & No & No & Yes & Yes & Yes & No & No & No & No & No & No \\ \hline
        \emph{Open Source} & Yes & No & Yes & Yes & No & Yes & Yes & Yes & Yes & Yes & Yes \\ \hline
        \emph{Crypto Primitive(s)} & AES, SHA & SHA & SHA & $-$ & $-$ & AES & SHA & AES & SHA & ECC, BLAKE2 & ECC \\ \hline
    \end{tabular}}
\label{tab:solutions_comparison}
\end{table*}

\textbf{GNSS Solutions.}
\ac{GNSS}-based protocols rely on geolocation data to perform proximity tracing. As such, there are several concerns regarding their accuracy, security, and privacy. First, it is worth noting that, with the exception of Corona 100m which is a centralized approach, \ac{GNSS} protocols offer perfect privacy to all citizens that do not contract \textsc{Covid-19}. This is because mobile devices do not send any information to the health authorities. However, for the individuals that test positive, these protocols reveal their (partial or full) trajectories over an extended period of time to the health authorities. As such, they are prone to linkage attacks that may cross-reference the published geolocation data with background information about the users. Even though the protocol by Berke \textit{et al.} employs hashing to obfuscate the trajectories, it is still vulnerable to brute-force attacks that can easily de-anonymize the locations.

In terms of security, \ac{GNSS} solutions are not susceptible to replay/relay and eavesdropping attacks, because they do not broadcast any data on the wireless channel~\cite{azad2020}. On the other hand, \ac{GNSS} signals can be spoofed by a malicious adversary, due to the lack of encryption/authentication in consumer \ac{GNSS} products. Such attacks may allow an adversary to manipulate the \ac{GPS} measurements that are collected by the contact tracing application. Another limitation of \ac{GNSS}-based contact tracing is that generic devices such as smartphones, are provided with a recreational grade (accuracy within $\pm 7.6m$) \ac{GNSS} receiver. This is not sufficient to determine ``real'' contagion events and may result in an overwhelming amount of false-negatives. Additionally, in an urban/sub-urban scenario, the shadowing and multipath fading effects further undermine the geolocation accuracy. Finally, the risk of \textsc{Covid-19} exposure is significantly higher in indoor environments where, unfortunately, \ac{GNSS} localization is infeasible.


\textbf{Centralized Solutions.}
Corona 100m is the only truly centralized protocol appearing in our literature review. This is not surprising, since centralized approaches come with significant security and privacy concerns. For example, in the case of Corona 100m, the app continuously reports its location to the public health authorities and is notified whether a confirmed case is within 100m of that location. As a result, every citizen that is actively using the app can be tracked by the authorities. Consequently, Corona 100m offers no privacy from the authority and it is also vulnerable to linkage attacks, because geolocation data can be cross-referenced with external sources to reveal sensitive user information. Even in an ideal world where the authorities are honest and trusted by the general public, the vast amount of valuable data stored at the centralized database serves as an open invitation to malicious hackers.

\textbf{Hybrid Solutions.}
The two notable protocols that fall into this category are \ac{PEPP-PT} and BlueTrace. Although the contact logs are stored in a decentralized manner (on the mobile devices) and are only disclosed if a user is diagnosed as positive, hybrid solutions can potentially result in a complete privacy breach. Recall that, under the hybrid model, the health authorities generate and distribute the ephemeral IDs to all devices, in order to perform the contact tracing task at their own server. As such, an eavesdropping adversary with access to the encryption keys that generate the ephemeral IDs can reconstruct the movements of every registered user.

Consequently, hybrid solutions have no privacy from authority and are extremely vulnerable to linkage and eavesdropping attacks. Also note that, the authority is not anonymizing the user identities, but is instead utilizing a pseudonym schema that can be vulnerable to several attacks described in the literature, such as \textit{single position attack, context linking attack, multiple position and context linking attack, multiple position attack}, and \textit{compromised trusted authority}~\cite{Wernke2012}. In other words, hybrid protocols demand a complete trust to the authorities on behalf of the users. Finally, both \ac{PEPP-PT} and BlueTrace incorporate timestamp data into the generated ephemeral IDs, which limits the effect of replay attacks (although the timing information is relatively coarse). On the other hand, they both fail to protect against relay attacks, since they do not take into account geolocation data.

\textbf{Decentralized Solutions.}
Here we focus on \ac{BLE}-based protocols that utilize symmetric key algorithms, including \ac{DP-3T}, Apple/Google, \ac{TCN}, and \ac{PACT}. As explained previously, their algorithms are quite similar and, therefore, they share the same security and privacy characteristics. First, similar to the decentralized \ac{GNSS}-based protocols (Hamagen and Berk \textit{et al.}), they offer perfect privacy to the vast majority of users who never contract the coronavirus. Furthermore, the beacons collected by an eavesdropping adversary cannot be used to de-anonymize them, because they never disclose their own beacons to the health authorities. Conversely, positively diagnosed users do not enjoy any privacy, because the disclosed ephemeral IDs can place them at precise locations by an eavesdropping adversary (or regular users who may trace those beacons inside their stored contact logs). As a result, decentralized protocols are not resilient against linkage attacks. In terms of replay attacks, all protocols use time as an input to the beacon generation process, thus reducing the risk of contact log manipulation by an active adversary. Note, however, that time is approximated in the scale of minutes, so there is still a sufficiently large time window that attackers can exploit. On the other hand, none of the aforementioned protocols utilize geolocation and they are, therefore, vulnerable to relay attacks.

A notable advantage of \ac{BLE} technology with regards to localization is its accuracy, which makes it an ideal candidate for effective contact tracing. Specifically, every transmitted beacon contains information about the time of arrival and the \ac{RSS}. The latter is a key parameter for wireless positioning and aids in deriving the distance between two Bluetooth enabled devices, by leveraging a radio propagation model~\cite{wang2013}. However, from a privacy perspective, it is crucial that the device's operating system implements a rolling MAC address protocol (e.g., by randomizing its last $3$ bytes), in order to prevent the tracking of individual users. It is also recommended that the schedule of the ephemeral IDs follows that of the MAC address randomization. A final remark concerns the security of the database server. Even though contact tracing is performed in a decentralized manner, the server plays a crucial role as a proxy for the distribution of the ``infected'' beacons. As such, the security of the health authorities' infrastructure is of paramount importance, even in the decentralized architecture~\cite{Ahmed2020_Access}.

\textbf{Public-key Solutions.}
Whisper is the only contact tracing protocol in the literature that employs public-key cryptography. Specifically, it allows any two mobile devices to exchange a secret contact ID (tell-token and hear-token) that is infeasible to compute without knowledge of the underlying public keys. As a result, when diagnosed users disclose their own tell-tokens to the centralized server, the authorities cannot map them to precise locations. In other words, Whisper offers excellent privacy from authority and is immune to eavesdropping attacks. However, the tokens that are stored by the mobile devices contain timing information, so individual mobile devices may be able to launch linkage attacks against the published tokens that are present in their memory. Another limitation of Whisper is that it does not employ any mechanism to mitigate replay/relay attacks. Finally, we should mention that Whisper is a complex protocol that does not rely simply on broadcasted beacons. Instead, to compute a new pair of tokens, the two devices have to establish a peer-to-peer connection and perform a Diffie-Hellman key exchange.

In the next section, we introduce \proto, a novel public-key based protocol that avoids the pitfalls of Whisper and is an excellent candidate for secure and privacy-preserving contact tracing.

\section{The SpreadMeNot Protocol}
\label{sec:proto}
In this section, we describe in detail our proposed solution for privacy-preserving contact tracing. In Section~\ref{sec:overview} we give a brief overview of the scheme and present the cryptographic experiment that will be the basis of our security proof. In Section~\ref{sec:construction} we introduce our construction and prove its security. In Section~\ref{sec:active} we propose a simple extension that mitigates several active attacks, and in Section~\ref{sec:properties} we analyze the privacy properties of our protocol in comparison to the existing state-of-the-art approaches.

\subsection{Overview}
\label{sec:overview}
The \proto\ protocol consists of a tuple of algorithms $\Pi = (\pcalgostyle{GenKey},\pcalgostyle{GenBeacon},\pcalgostyle{RandBeacon},\pcalgostyle{Test})$ that operate as follows.
\begin{enumerate}\itemsep6pt
    \item \pcalgostyle{GenKey}: It takes as input a security parameter $n$ and outputs a private key $x$ and a public-key $P$.
    \item \pcalgostyle{GenBeacon}: It takes as input a public-key $P$ and outputs a beacon $\mathcal{C}$.
    \item \pcalgostyle{RandBeacon}: It takes as input a beacon $\mathcal{C}$ and outputs a randomized version $\mathcal{C}'$.
    \item \pcalgostyle{Test}: It takes as input a beacon $\mathcal{C}$ and a private key $x$, and outputs \textit{true} if the beacon is generated under public-key $P$.
\end{enumerate}

Initially, every user executes the \pcalgostyle{GenKey} algorithm to generate their public/private key pair. However, in this setting, the public-key is irrelevant to the rest of the users and does not need to be shared (but it does not have to be secret). At regular intervals (e.g., every $15$ minutes), the user's app will invoke the \pcalgostyle{GenBeacon} algorithm to generate a \textit{fresh} random beacon. That beacon will be broadcast continuously until the next interval, in order for the surrounding devices to detect the user as a potential contact. Note that only \emph{significant} events are stored on each device, i.e., beacons that are very close to the user (say, within $2$ meters) for a sufficiently long time duration. Based on the characteristics of the virus, old contact events (e.g., older than two weeks) are automatically erased. 

If a user tests positive for  \textsc{Covid-19}, he publishes his contact list on the centralized database managed by the public health authorities. The list consists of all the beacons that are currently stored on his device. More importantly, the published list is first \textit{permuted} and all its entries are randomized. Specifically, for each stored beacon $\mathcal{C}$, algorithm \pcalgostyle{RandBeacon} is invoked to produce a new beacon $\mathcal{C}'$ that is indistinguishable from $\mathcal{C}$. Finally, the published contact list is downloaded by all mobile devices that subsequently apply algorithm \pcalgostyle{Test} on every beacon. If any instance of \pcalgostyle{Test} outputs \textit{true}, the user infers that he has been in close proximity to the infected individual.

In terms of security, we want beacons generated by different public keys to be indistinguishable from each other. Therefore, we define an eavesdropping adversary $\adv$ that has access to all transmitted and (randomized) published beacons. The objective of $\adv$ is to distinguish beacons that are generated from a specific public-key $P$. Experiment $\pcalgostyle{Exp}^{\pcalgostyle{eav}}_{\adv,\Pi}(n)$, as shown in Figure~\ref{fig:game}, formulates the security of our protocol $\Pi = (\pcalgostyle{GenKey},\pcalgostyle{GenBeacon},\pcalgostyle{RandBeacon},\pcalgostyle{Test})$ with security parameter $n$, under an eavesdropping adversary $\adv$. We say that $\adv$ succeeds in this experiment if
$$\text{Pr}[\pcalgostyle{Exp}^{\pcalgostyle{eav}}_{\adv,\Pi}(n)=1] \le \frac{1}{2}+ \negl$$
Conversely, our scheme is secure if $\adv$ fails in this experiment.

\begin{figure}[ht]
\centering
\fbox{\begin{varwidth}{\dimexpr\columnwidth-3\fboxsep-3\fboxrule\relax}
\begin{center}
	\vspace{10pt}
\textbf{Beacon Indistinguishability\\Experiment} $\pcalgostyle{Exp}^{\pcalgostyle{eav}}_{\adv,\Pi}(n)$
\end{center}
\begin{enumerate}\itemsep6pt
    \item On input a security parameter $n$, \pcalgostyle{GenKey} generates two private keys $x_0,x_1$ and the corresponding public keys $P_0,P_1$; the public keys are given to $\adv$.
    \item $\adv$ calls \pcalgostyle{GenBeacon} an arbitrary number of times.
	\item The challenger selects a random bit $b\sample \bin$.
	\item If $b=0$, the challenger uses \pcalgostyle{GenBeacon} to generate a beacon $\mathcal{C}_0$, under public-key $P_0$. If $b=1$, the challenger generates $\mathcal{C}_1$ under public-key $P_1$. 
	\item $\adv$ continues calling \pcalgostyle{GenBeacon} and eventually outputs a bit $b'$.
	\item Return 1 if $b=b'$ and 0 otherwise.
	\vspace{10pt}
\end{enumerate}
\end{varwidth}
}
\caption{Beacon indistinguishability experiment.}
\label{fig:game}
\end{figure}

\subsection{Construction and Security Proof}
\label{sec:construction}
Assume an elliptic curve group $\GG$ of prime order $q$ and let $G$ be a generator of $\GG$. This information is public and shared by all users. Also, $n$ is the security parameter that determines the precise construction of the group $\GG$. Our protocol $\Pi$ is instantiated as follows.
\begin{enumerate}\itemsep6pt
    \item \pcalgostyle{GenKey}: On input a security parameter $n$, choose a uniformly random private key $x\in \ZZ^*_q$ and set the public-key $P = xG$.
    \item \pcalgostyle{GenBeacon}: On input a public-key $P$, choose a uniformly random $r\in \ZZ^*_q$ and output beacon $\mathcal{C} = (rG,rP)$.
    \item \pcalgostyle{RandBeacon}: On input a beacon $\mathcal{C} = (rG,rP)$, choose a uniformly random $r'\in \ZZ^*_q$ and output $\mathcal{C}' = (r'rG,r'rP) = (r''G,r''P)$.
    \item \pcalgostyle{Test}: On input a beacon $\mathcal{C} = (rG,rP)$ and a private key $x$, compute $A=xrG$. If $A=rP$, output \textit{true}; otherwise, output \textit{false}.
\end{enumerate}

It can be easily verified  that the protocol is \textit{correct}, i.e., if algorithm \pcalgostyle{Test} returns \textit{true}, the input beacon is generated under public-key $P$, since $xG=P$. The following theorem proves the security of our scheme.

\begin{theorem}
There is no polynomial-time adversary $\adv$ that succeeds in $\pcalgostyle{Exp}^{\pcalgostyle{eav}}_{\adv,\Pi}(n)$.
\end{theorem}
\begin{proof}
The only information that the adversary has is a beacon $\mathcal{C} = (rG,rP)$. If $\adv$ succeeds in the experiment, it means that it can distinguish between $rx_0G$ and $rx_1G$ with non-negligible probability. Note that, it is infeasible to derive any of the secret values $r, x_0, x_1$, due to the \ac{ECDLP} assumption. Since $r, x_0, x_1$ are uniformly random in $\ZZ^*_q$, both $rx_0G$ and $rx_1G$ would appear as random elements in $\GG$. The probability of distinguishing between any two random group elements is negligible in $\log{q}$, where $q$ is the order of $\GG$ and $\log{q}$ is the order of the security parameter $n$. Since $\adv$ is polynomial-time, the experiment succeeds with probability $\negl$ larger than a random guess, which concludes our proof.
\end{proof}

\textcolor{black}{As shown in Fig.~\ref{fig:protocol}, to construct the beacon at time slot $k$, user $i$ chooses a random value $r_k$, such that $r_k < q$, and computes the tuple $\langle r_kG, r_kP_i\rangle$ 
that is broadcast by user $i$ at time slot $k$.} 
\textcolor{black}{As depicted in Fig.~\ref{fig:protocol_verify}, when a user $j$ downloads a new contact list, he must check whether any of his own beacons are included in that list. This is  done as follows: for every tuple $\langle rG, rP\rangle$ in the contact list, the user multiplies the first term with his own secret key $x_j$. If the result is equal to the second term, the user confirms that he is the owner of that particular beacon. It is important to note that the user cannot compare any other beacon against his own list and, thus, can not deduce any information about the person(s) in the published list.}

\begin{figure}[ht!]
    \centering
    \begin{tikzpicture}[node distance=4cm, people/.style={minimum width=1.0cm}, auto]
    \node[people, alice] (alice) {User $i$};
    \node[people, bob, right=of alice] (bob) {User $j$};
    \draw[] ([yshift=-1cm]alice.south) coordinate (l1)--(l1-|alice) node[midway, above]{$(x_i, P_i = x_iG)$};
    \draw[] ([yshift=-1.5cm]alice.south) coordinate (l1)--(l1-|alice) node[midway, above]{Choose $r_k <q$};
    \draw[] ([yshift=-1cm]bob.south) coordinate (l1)--(l1-|bob) node[midway, above]{$(x_j, P_j = x_jG)$};
    \draw[->,dashed] ([yshift=-2.5cm]alice.south) coordinate (l1)--(l1-|bob) node[midway, above]{$(r_kG, r_kP_i)$};
    \draw [black, ->] ([xshift=-1.5cm] 0,-1) -- ([xshift=-1.5cm] 0,-3) node [midway, rotate=90, fill=white, yshift=2.5pt, above] {time slot $k$};
    \end{tikzpicture}
    \caption{Sequence Diagram of the \proto\ scheme.}
    \label{fig:protocol}
\end{figure}
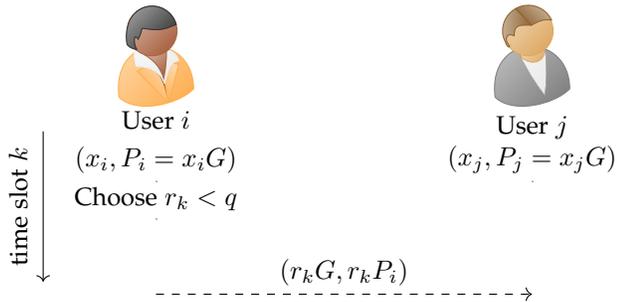

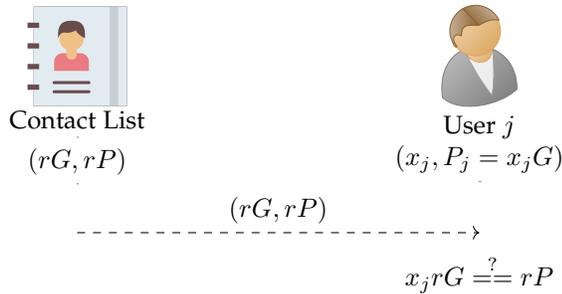
\begin{figure}[ht!]
    \centering
    \begin{tikzpicture}[node distance=4cm, people/.style={minimum width=1.0cm}, auto]
    \node[] (cl) {\includegraphics[width=.08\textwidth]{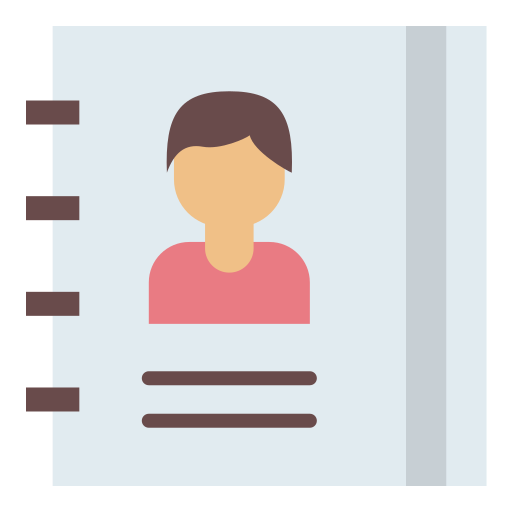}};
    \node[people, bob, right=of cl] (bob) {User $j$};
    \draw[] ([yshift=-0.23cm]cl.south) coordinate (l1)--(l1-|cl) node[midway, above]{Contact List};
    \draw[] ([yshift=-0.85cm]cl.south) coordinate (l1)--(l1-|cl) node[midway, above]{$(rG,rP)$};
    \draw[] ([yshift=-1cm]bob.south) coordinate (l1)--(l1-|bob) node[midway, above]{$(x_j, P_j = x_jG)$};
    \draw[->,dashed] ([yshift=-1.5cm]cl.south) coordinate (l1)--(l1-|bob) node[midway, above]{$(rG, rP)$};
     \draw[] ([yshift=-2.6cm]bob.south) coordinate (l1)--(l1-|bob) node[midway, above]{$x_jrG  \stackrel{?}{==} rP$};
    \end{tikzpicture}
    \caption{\proto\ check Contact List procedure.}
    \label{fig:protocol_verify}
\end{figure}

\subsection{Mitigating Active Attacks}
\label{sec:active}
The downside of using public-key cryptography is that it facilitates impersonation attacks. That is, if an adversary has knowledge of a user's public key, he can trivially generate valid beacons on his behalf. Even worse, knowledge of a user's public key is not really necessary; instead, the adversary may intercept a single beacon and then use algorithm \pcalgostyle{RandBeacon} to generate new, valid beacons at will. In addition to impersonation attacks, it is also essential to protect against other active attacks, including replay and relay attacks. 

To this end, we propose a simple solution, based on timestamps, localization data, and digital signatures, to mitigate such attacks. First, we assume that the users' mobile devices are synchronized to within a few seconds and every device knows its approximate location (longitude and latitude). Then, every beacon is transmitted with an attached Unix timestamp, which represents the current time, and the user's current coordinates. In addition, the beacon includes a digital signature that is computed as follows.
\begin{enumerate}
    \item Let $\mathcal{C} = (rG,rP)$ be the current beacon, let $T$ be the corresponding Unix timestamp, and let $L$ be the user's current location.
    \item Let $d = rx \bmod q$ be the private key associated with public-key $P' = rP = rxG$, i.e., the second term of the beacon.
    \item Use the \ac{ECDSA} algorithm to generate the signature $\sigma$ of message $L||T||rG||rP$, under private key $d$.
    \item Broadcast beacon $\mathcal{C} = (L,T,rG,rP,\sigma)$.
\end{enumerate}

At the receiver side, the device will first verify that $T$ is within the synchronization threshold and $L$ sufficiently close. Then, it will use the public-key $P' = rP$ to verify signature $\sigma$. If any of these tests fail, the beacon will be discarded. As such, any attempt from an adversary to re-use that beacon (or a randomized version of it) at subsequent timestamps (or remote locations) will fail. Note that, to save valuable resources, it is not necessary to verify a signature as soon as it is received. Recall that a beacon is only stored in the contact list if it has been received multiple times over a sufficiently long time interval. Therefore, the device can simply defer the signature verification process until a beacon is inserted into the contact list.

\subsection{Security and Privacy Analysis}
\label{sec:properties}
We now discuss the security and privacy characteristics of our protocol, with respect to Table~\ref{tab:solutions_comparison}.
In terms of privacy, \proto\ has several clear advantages over the current state-of-the-art approaches. First, it is immune to eavesdropping attacks, because the published information (from infected users) is randomized and, thus, cannot be linked to any previously intercepted beacon by an adversary. As such, \proto\ also protects the privacy of the users who have been positively diagnosed. Second, the infected user never reveals his/her own beacons and, therefore, cannot be traced inside a malicious user's contact list, which may include detailed timestamp and location information for each contact. More importantly, the published contact lists are permuted and do not include any identifiable information besides the two elliptic curve points (i.e., the signatures and all metadata are removed). As a result, even if an adversary matches a beacon inside the published contact list, it is impossible to tie that beacon to a specific location and timestamp. Consequently, our protocol is very resilient against linkage attacks. Finally, by attaching a digital signature in every transmitted beacon, \proto\ is the first protocol in the literature that mitigates both replay and relay attacks.

\section{Performance Evaluation}
\label{sec:performance}
\textcolor{black}{In this section we provide an experimental evaluation of the \proto\ protocol. It is worth noticing that, compared to the aforementioned solutions proposed in the literature, \proto\ is not the most 
computational 
or energy efficient solution---while still being largely viable. But, \proto\ sports provable security and privacy features unmatched by competing solutions.
Indeed, we want to remark that \proto\ is  a trade-off approach that mitigates the main security and the privacy issues that emerged by our analysis of the existing contact tracing solutions, at the expense of a slight increase in the energetic and computational overhead. 
Further possible performance improvements 
are left as future work.}

We evaluated the cryptographic component of \proto\ on an LG Google Nexus 5X, equipped with a hexa-core $64$-bit CPU ($4\times1.4$~GHz Cortex-A53 and $2\times1.8$~GHz Cortex-A57) ARMv8-A. The wireless communications module was evaluated on a Qualcomm Atheros QCA6174A SoC hardware platform with built-in IEEE 802.11ac radio, Bluetooth v4.2, and \acl{BLE}. The device supports the standard 256-QAM modulation and utilizes $1,216$~KB RAM and $448$~KB ROM for Wi-Fi, and $192$~KB RAM and $672$~KB ROM for Bluetooth~\cite{qualcomm_chip_radio}.

In particular, the cryptographic operations of the \proto\ protocol where implemented in \emph{C++}, with the well known \emph{OpenSSL 1.1.1d} C library~\cite{openssl}, using Android Native Development Kit (NDK) with \emph{Android 8.1 Oreo} OS~\cite{android}. 
For the experimental evaluation, we selected four elliptic curves, i.e., \pcalgostyle{secp128r1,\ secp160r1,\ secp192k1}, and \pcalgostyle{secp256k1}. According to NIST's latest guidelines, these curves provide security levels of $64$, $80$, $96$, and $128$~bits, respectively~\cite{Barker2020}. Furthermore, following NIST's recommendations, we adopted the \pcalgostyle{SHA256} hash function for the digital signatures. We considered the following two performance metrics for our protocol: (i) CPU time for the beacon generation function; and, (ii) energy consumption for the computations and TX/RX operations.

Specifically, to estimate the energy consumption due to the public-key cryptographic primitives, we focused on the protocol's basic operation---the elliptic curve point-scalar multiplication---which is by far the most CPU intensive operation. To this end, we leveraged the power profile component of the Android OS that outputs the current consumption values for sensors and CPUs and the approximate battery drain caused by the component over time. The power profile is provided by the device manufacturer~\cite{android_power_values}. After performing the scalar-point multiplication, we estimated the current drained at $105.87$~mA. Therefore, we concluded that this value is the instantaneous current drained by the CPU to perform the elliptic curve operation. Additionally, we measured the battery voltage at $3.9$~V, using a common battery monitoring application.

Fig.~\ref{fig:ecc_time} illustrates the CPU time that is required to perform a single scalar-point multiplication for various security levels. At $80$ bits security, the operation takes $1.1196$~ms to complete and it consumes $\approx 0.4623$~mJ of energy. When we require $128$ bits security, the cost is approximately doubled, i.e., $2.4030$~ms of CPU time and $\approx 0.9922$~mJ of consumed energy. The tests related to the CPU time were repeated $5,000$ times and, in Fig.~\ref{fig:ecc_time}, we report the mean value and the $95\%$ confidence intervals.

\begin{figure}[htbp]
    \centering
    \includegraphics[width=\columnwidth]{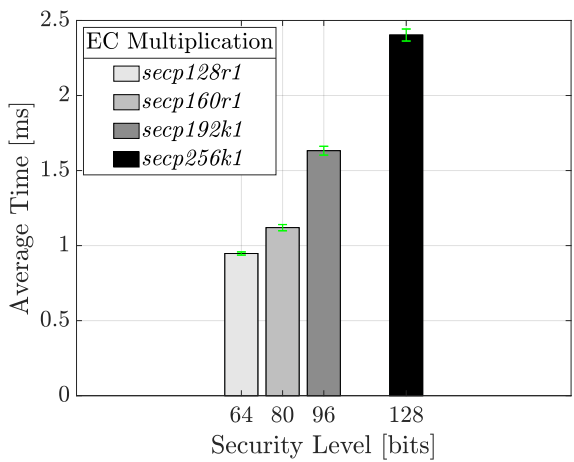}
    \caption{Average CPU time for elliptic curve point-scalar multiplication.}
    \label{fig:ecc_time}
\end{figure}


To measure the energy consumption and transmission time during the beacon broadcast operation, we assumed the Bluetooth v4.2 standard, where the data transmission rate is $1$~Mbps. Nevertheless, given the multipath and slow fading effects that negatively affect RF communications (due to scattering, reflection, and diffraction), we adopted a conservative stance and assumed that the maximum throughput is $0.8$~Mbps. We evaluated \proto\ at the MAC layer, where the Bluetooth v4.2 standard sets the maximum frame size at $265$ bytes (using packet length extension), which translates to a maximum payload size of $251$ octets \cite{bluetooth_info}. Thus, if the protocol needs to transmit larger payloads, they must be fragmented. Note that, the amount of transferred data is related to the desired security level, i.e., the communication cost increases for more secure elliptic curve groups. However, one technique to reduce the payload size is to employ point compression, where only one coordinate is transmitted for each elliptic curve point. 

The energy consumption of the TX and RX operations for the Qualcomm Atheros QCA6174A SoC module can be computed from the area underlying the current consumption curve, as depicted in Eq.~\ref{eq:energy}.
\begin{equation}
    \label{eq:energy}
    E [mJ] = 1.8V \cdot \int_{0}^{\tau} i(t)dt
\end{equation}
Here, $E$ is the energy consumption (measured in mJ), $i(t)$ is the instantaneous current drain (in mA), $\tau$ is the operation duration, and $1.8$~V is the minimum voltage of the QCA6174A board (in Volts). This translates to $\approx46$~mA in TX mode and $\approx42$~mA in RX mode, for the Bluetooth v4.2 protocol.

Fig.~\ref{fig:test_time_energy_consumption} illustrates the total energy consumption and the time required to complete one beacon exchange under the \proto\ protocol. First, the energy consumption ranges from a minimum of $2.552$~mJ to a maximum of $6.260$~mJ, based on the underlying security level. Similarly, for $64$ bits security, the average time to complete the exchange is $8.274$~ms, where $5.6837$~ms are required for the cryptographic operations (i.e., beacon generation, and signature generation and verification), while $2.5903$~ms are devoted to the exchange of the crypto material. For $128$ bits security, the time to complete the protocol increases to $18.288$~ms, where $14.418$~ms are devoted to the cryptographic computations and $3.87$~ms are related to the exchange of the crypto material. We emphasize that the overall duration of the beacon exchange can be affected by the configuration of the operating system features at the MAC link-layer, and also various RF phenomena at the PHY layer.

\begin{figure}[htbp]
    \centering
    \includegraphics[angle=0, width=\columnwidth]{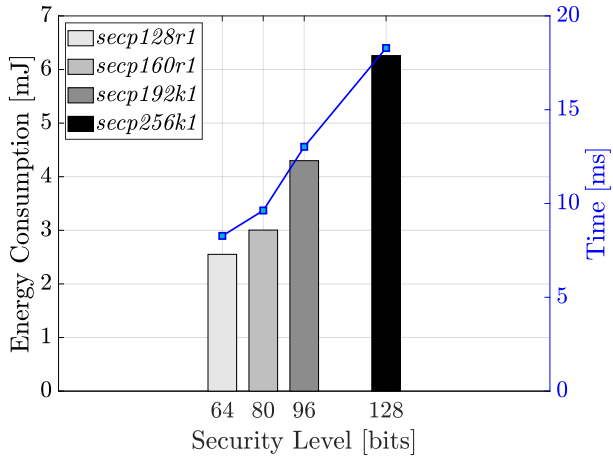}
    \caption{Total time and energy consumption (TX, RX, and Idle Status) of \proto\ on the LG Google Nexus 5X smartphone.}
    \label{fig:test_time_energy_consumption}
\end{figure}

Assuming the standard-compliant MAC frames in Bluetooth v4.2, Table~~\ref{tab:comparison} summarizes the different cost associated with a single beacon exchange in the \proto\ protocol. The results confirm that, even at the highest security level, \proto\ is feasible for resource-constrained smart devices. Indeed, the battery capacity of the LG Google Nexus 5X smartphone is $41,796$~J ($2,700$~mAh) so, for $128$ bits security, our protocol consumes $\approx 1.5\cdot10^{-7}$\% of the battery capacity. Also note that the payload size is well within the standard's limit, so beacons can be transmitted within a single MAC frame.

\begin{table}[htbp]
\color{black}
\caption{\proto\ cost summary for one beacon exchange.
}
\centering
\begin{tabular}{|c|c|c|c|c|}
\hline
\multirow{2}{*}{\textbf{Feature}} & \multicolumn{4}{c|}{\textbf{Security Level (bits)}} \\ \cline{2-5} 
 & \textbf{\makecell[c]{$\mathbf{64}$}} & \textbf{\makecell[c]{$\mathbf{80}$}} & \textbf{\makecell[c]{$\mathbf{96}$}} & \textbf{\makecell[c]{$\mathbf{128}$}} \\ \hline
\textit{Payload Size (B)} & $68$ & $84$ & $100$ & $132$ \\ \hline
\textit{Energy Consumption (mJ)} & $2.552$ & $3.004$ & $4.299$ & $6.260$ \\ \hline
\textit{Time Duration (ms)} & $8.274$ & $9.628$ & $13.023$ & $18.288$  \\ \hline
\end{tabular}
\label{tab:comparison}
\end{table}

\section{Discussion}
\label{sec:discussion}
Our experimental evaluation detailed in the previous section illustrates the feasibility of public-key cryptography in the context of contact tracing. Compared to the state-of-the-art symmetric key approaches, such as Apple/Google, \proto\ shows a more sustained overhead, though compensated from much stronger security and privacy guarantees. In our evaluation, we focused on the beacon exchange process, which is the most frequently performed operation, like for the vast majority of contact tracing apps. 
However, one operation where the performance of \proto\ is orders of magnitude more expensive with respect to other solutions is the exposure notification operation. Though, it should be noted that this operation is invoked much less frequently, and could be easily outsourced, as discussed later. 

In detail, in the event of a positive diagnosis, the app has to re-randomize all beacons in its contact list prior to uploading them to the centralized server. Assuming an average of $200$ significant contact events per day, over a period of $14$ days, this amounts to $2,800$ entries. Recall that the \pcalgostyle{RandBeacon} function necessitates two point-scalar multiplications and, based on the results of Fig.~\ref{fig:ecc_time}, the entire process may consume (for $128$ bits security) up to $13.44$~s of computing time and $5.55$~J of energy. 
Additionally, the communication cost to upload the beacons is $179.2$~KB. Again, these costs may be dramatically reduced: cut by half if we choose $80$ bits security, or completely removed if we resort to outsourcing, as discussed in the following.


An even more resource-demanding operation is the actual contact tracing, where individual devices have to identify possible contagion events, after downloading the latest contact lists coming from recently diagnosed patients. Assuming that this operation is performed on a daily basis, the typical number of beacons that need to be matched would be in the order of millions (e.g., $5.6$ million, if there are $2,000$ new cases with $2,800$ stored beacons each). In this example, the smartphone would consume (for $128$ bits security) approximately $3.7$~h of computing time and $5,556$~J of energy. There is also a communication cost of $360$~MB to download the beacons from the centralized server. Again, these costs may be dramatically reduced: cut by half if we choose $80$ bits security, or completely removed if we resort to outsourcing, as discussed in the following.

Despite the highlighted  cost, we strongly believe that \proto's superior security and privacy properties offset any argument regarding its performance. More importantly, the cited costs are easily mitigated by separating the \textit{online} and \textit{offline} components of the protocol. Specifically, it is worth noting that all the point-scalar multiplications in the beacon generation process can be performed offline. There are a total of three such operations: two for computing the beacon itself and one for computing the digital signature (which is input-independent). Therefore, we can pre-compute offline a large number of random elliptic curve points that may be used during the online phase (beacon exchange) by the device. Consequently, the beacon generation process can be reduced to computing just one hash function and two integer multiplications modulo $q$. Both are very cheap operations---involved in the computation of the digital signature. 

Similarly, all the remaining expensive operations (contact list maintenance and contact tracing) may be safely moved into the offline module. Contact list maintenance involves the verification of all signatures for the beacons that are inserted into the list (to detect replay/relay attacks), and the re-randomization of the accepted beacons in order to minimize the online computational cost in the event of a positive diagnosis. More importantly, contact tracing is also independent of the protocol's everyday operations and may be performed in an offline fashion. To this end, we propose the following two approaches for handling the offline operations:
\begin{itemize}
    \item \textbf{Night mode:} The offline work will be performed during night time, when the smartphone is charging and connected to the home WiFi network.
    \item \textbf{Desktop app:} A companion desktop app will be developed to handle the offline tasks. The user will periodically sync the smartphone with the desktop app, in order to upload and download the necessary information. Note that this one would be the preferred solution, as modern desktop CPUs can perform the cryptographic operations significantly faster. Also, the app will implement multi-threading to take advantage of the multiple CPU cores, since all the protocol's operations are highly parallelizable.
\end{itemize}





\section{Conclusion}
\label{sec:conclusion}
In this paper, we have provided two main  contributions in the domain of digital contact tracing.
We first 
performed a thorough evaluation of the privacy and security characteristics of the main contact tracing apps, focusing on their basic mechanisms. 
Our results showed that the most prominent solutions fail to protect the privacy of positively diagnosed individuals, as well as being 
vulnerable to a variety of active and passive attacks. 
To cope with the above-highlighted shortcomings, our second  contribution was the design of \proto, a novel, provably secure contact tracing protocol that protects the privacy of \textit{all} users while at the same time being resilient against most passive and active attacks, including eavesdropping and replay/relay attacks. Although \proto\ is based on public-key cryptographic primitives that are computationally expensive, we experimentally showed that the induced overhead 
can be easily handled by modern smartphones. 
In addition, we have demonstrated that all resource-intensive operations can be safely moved into an offline module, thus rendering \proto's online proximity tracing task extremely lightweight. Given the strong and provable security and privacy properties, combined with the experimentally proved sustainable overhead, 
we argue that \proto\ is the ideal candidate for digital contact tracing. Finally, the open source nature of this project will also pave the way for further solutions in the domain.




\section*{Acknowledgements}
\label{sec:ack}
This publication was partially supported by awards NPRP11S-0109-180242
from the QNRF-Qatar National Research Fund, a member of The Qatar Foundation. The information and views set out in this publication are those of the authors and do not necessarily reflect the official opinion of the QNRF.

\bibliographystyle{IEEEtran}
\bibliography{tracing}
\section*{Biographies}
\noindent

\begin{IEEEbiography}[{\includegraphics[width=1in,height=1.25in,clip,keepaspectratio]{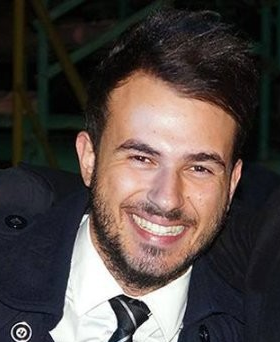}}]{Pietro Tedeschi} is currently a PhD Student in Computer Science and Engineering (Cybersecurity) at the Hamad Bin Khalifa University (HBKU), Doha, Qatar. He is an active member of the HBKU Cyber-Security Research Innovation Lab. He received his Bachelor's degree in Computer and Automation Engineering in 2014 with a thesis on the Analysis of Security Protocols for the Internet of Things, in IEEE 802.15.4e Networks, and his Master's degree (with honors) in Computer Engineering both from the ``Politecnico di Bari``, in 2017 with a thesis on the Development of Security Architectures in Intelligent Transport Systems for EU Horizon 2020 BONVOYAGE project. From 2017 to 2018, he worked as Security Researcher at CNIT (Consorzio Nazionale Interuniversitario per le Telecomunicazioni), Italy, for the EU H2020 SymbIoTe project. His research interests span over UAV/Drone Security, Wireless Security, Internet of Things (IoT), Applied Cryptography, and Cyber-Physical Systems.
\end{IEEEbiography}

\begin{IEEEbiography}[{\includegraphics[width=1in,height=1.25in,clip,keepaspectratio]{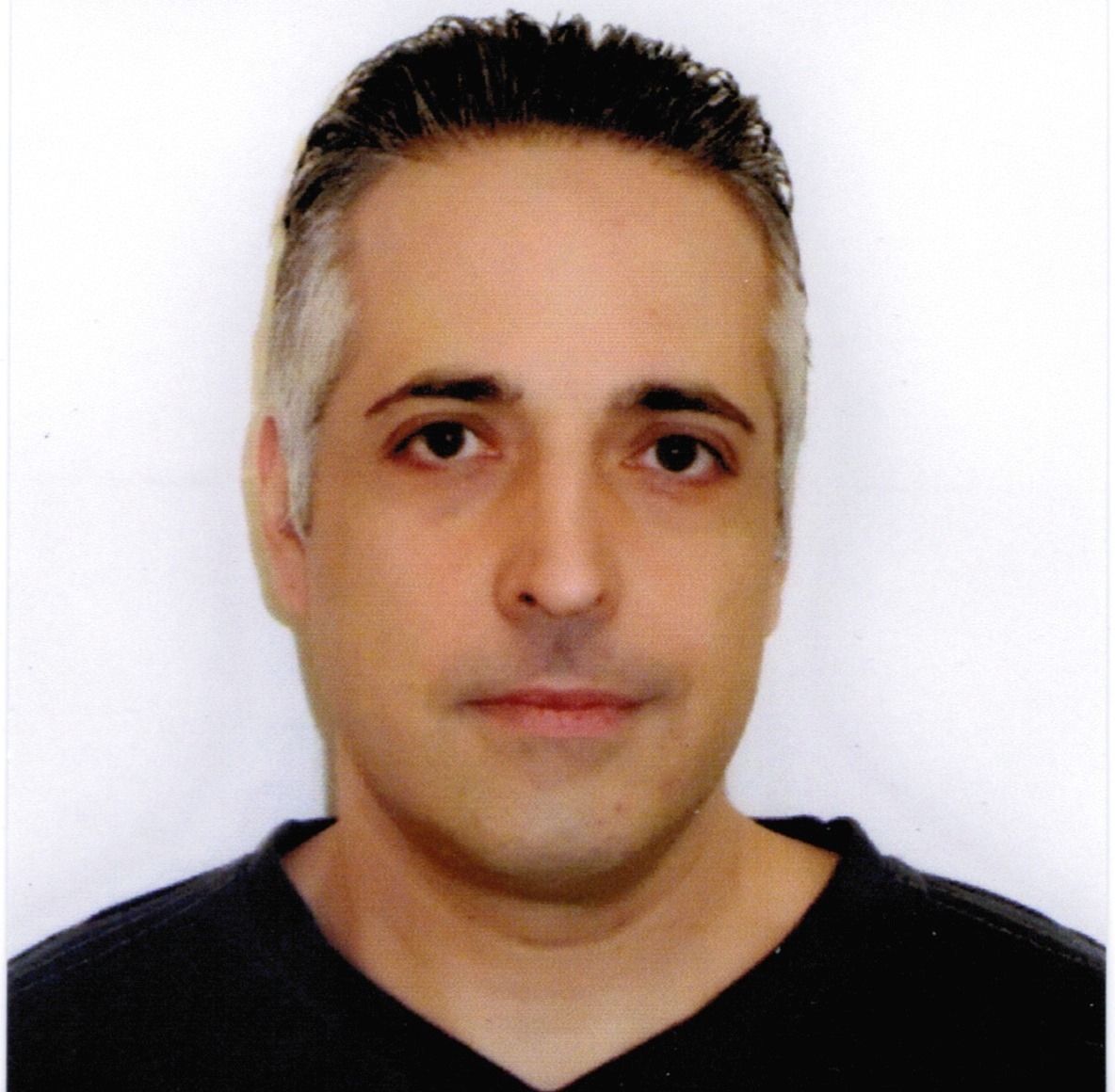}}]{Spiridon Bakiras} received the B.S. degree in electrical and computer engineering from the National Technical University of Athens, in 1993, the M.S. degree in telematics from the University of Surrey, in 1994, and the Ph.D. degree in electrical engineering from the University of Southern California, in 2000. He is currently an Associate Professor with the College of Science and Engineering, at Hamad Bin Khalifa University, Qatar. Before that, he held teaching and research positions at Michigan Technological University, The City University of New York, The University of Hong Kong, and The Hong Kong University of Science and Technology. His current research interests include database security and privacy, mobile computing, and spatiotemporal databases. He is a member of the ACM and IEEE, and a recipient of the U.S. National Science Foundation (NSF) CAREER Award. 
\end{IEEEbiography}

\begin{IEEEbiography}[{\includegraphics[width=1in,height=1.30in,clip,keepaspectratio]{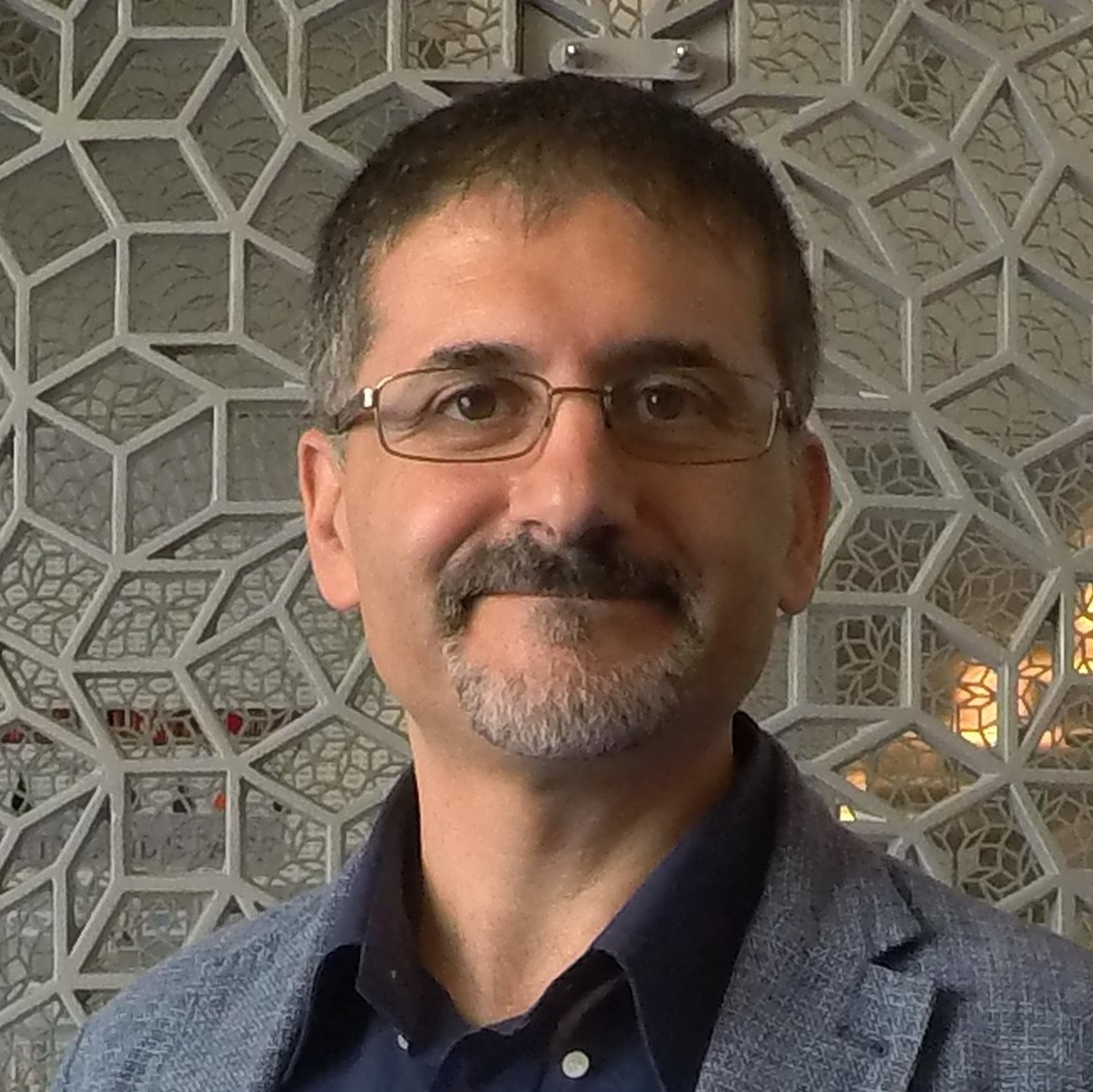}}]{Roberto Di Pietro, } 
 ACM Distinguished Scientist, is Full Professor in Cybersecurity at HBKU-CSE. Previously, he was in the capacity of Global Head Security Research at Nokia Bell Labs, and Associate Professor (with tenure) of Computer Science at University of Padova, Italy. He also served 10+ years as senior military technical officer. Overall, he has been working in the cybersecurity field for 23+ years, leading both technology-oriented and research-focused teams in the private sector, government, and academia (MoD, United Nations HQ, EUROJUST, IAEA, WIPO). His main research interests include security and privacy for wired and wireless distributed systems (e.g. Blockchain technology, Cloud, IoT, On-line Social Networks), virtualization security, applied cryptography, computer forensics, and data science. 
Other than being involved in M\&A of start-up---and having founded one (exited)---, he has been producing 230+ scientific papers and patents over the cited topics, has co-authored three books, edited one, and contributed to a few others. 
He is serving as an AE for ComCom, ComNet, PerCom, Journal of Computer Security, and other Intl. journals. 
In 2011-2012 he was awarded a Chair of Excellence from University Carlos III, Madrid. In 2020 he received the Jean-Claude Laprie Award for having significantly influenced the theory and practice of Dependable Computing. 
\end{IEEEbiography}

\end{document}